%% file: BlockchainMiningGames.tex
\documentclass[11pt]{article}
\usepackage{amsmath,amsfonts,amsthm,amssymb}
\usepackage[a4paper,margin=1in]{geometry}
\usepackage{graphicx}
\usepackage{subcaption}
\usepackage{todonotes}
\usepackage{eurosym}
\theoremstyle{plain}

\newtheorem{theorem}{Theorem}
\newtheorem{corollary}{Corollary}
\newtheorem{lemma}{Lemma}
\newtheorem{proposition}{Proposition}

\newtheorem{definition}{Definition}

\newtheorem{claim}{Claim}

\newcommand{\f}{\varphi}
\newcommand{\ff}{\hat\varphi}
\newcommand{\sol}{\overline\varphi}
\newcommand{\g}{g^*}
\newcommand{\gx}{\hat g}
\newcommand{\gxs}{\hat g^*}
\newcommand{\frontier}{\textsc{Frontier}}

\def\bitcoin{%
  \leavevmode
  \vtop{\offinterlineskip 
    \setbox0=\hbox{B}%
    \setbox2=\hbox to\wd0{\hfil\hskip-.03em
    \vrule height .3ex width .15ex\hskip .08em
    \vrule height .3ex width .15ex\hfil}
    \vbox{\copy2\box0}\box2}}

\makeatletter
\newcommand{\tagL}{\tagsleft@true}
\newcommand{\tagR}{\tagsleft@false}
\makeatother

\title{Blockchain Mining Games\footnote{This work is supported by ERC Advanced Grant 321171 (ALGAME), by ERC project CODAMODA, \# 259152,  and by H2020 project PANORAMIX, \# 653497.}}

\author{Aggelos Kiayias\thanks{University of Edinburgh (most of the work completed while at the National and Kapodistrian University of Athens). Email: \texttt{akiayias@inf.ed.ac.uk}}
\and
Elias Koutsoupias
\thanks{University of Oxford. Email: \texttt{\{elias,kyropoul\}@cs.ox.ac.uk}}\and
Maria Kyropoulou$^\ddag$ \and
Yiannis Tselekounis\thanks{National and Kapodistrian University of Athens. Email: \texttt{tselekounis@di.uoa.gr}}
}

\date{\empty}

\begin{document}
\maketitle

\begin{abstract}
  We study the strategic considerations of miners participating in the
  bitcoin's protocol. We formulate and study the stochastic game that
  underlies these strategic considerations. The miners collectively
   build a tree of blocks, and they are paid when they create a node (mine a block)
  which will end up in the path of the tree that is adopted by all.
  Since the miners can hide newly mined nodes, they play a game with incomplete information. Here we
  consider two simplified forms of this game in which the miners have
  complete information. In the simplest game the miners release every
  mined block immediately, but are strategic on which blocks to
  mine. In the second more complicated game, when a block is mined it
  is announced immediately, but it may not be released so that other
  miners cannot continue mining from it. A miner not only decides
  which blocks to mine, but also when to release blocks to other
  miners. In both games, we show that when the computational power of
  each miner is relatively small, their best response matches the
  expected behavior of the bitcoin designer. However, when the
  computational power of a miner is large, he deviates from the
  expected behavior, and other Nash equilibria arise.
  \end{abstract}

\input{intro.tex}
\input{model.tex}
\input{immediate_release.tex}

\input{strategic_release.tex}
\bibliographystyle{apalike}
\bibliography{bitcoin-bibfile}

\end{document}

%% file: intro.tex
\section{Introduction}

Bitcoin is the most successful decentralized digital currency. It was first presented in the white paper \cite{Nakamoto08} under the pseudonym Satoshi Nakamoto. Its backbone is the blockchain protocol which attempts to keep a consisted list of transactions in a peer-to-peer network. The blockchain protocol successfully solves the real distributed problem of agreement, and has the potential to support novel applications which require distributed computing across a network.

Game-theoretic issues are very important for the correct execution of the blockchain protocol. This was realized at its inception when its creator, Nakamoto, analyzed incentives in a simple, albeit insufficient, model. Understanding these issues is essential for the survival of bitcoin and the development of the blockchain protocol. In practice it can help understand their strengths and vulnerabilities and, in economic and algorithmic theory, it can provide an excellent example for studying how rational (selfish) players can play games in a distributed way and map out their possibilities and difficulties.

Distilling  the essential game-theoretic properties of blockchain maintenance  is far from trivial; some ``attacks'' and vulnerabilities have been proposed but there seems to exist no systematic way to discover them. In this work, we study two models in which the miners (the nodes of the distributed network that run the protocol and are paid for it) play a \emph{complete-information stochastic game}. Although the miners in the actual blockchain game do not have complete information, our games aim to capture two important questions that selfish miners ask: (a) what to compute next (more precisely, which block to mine) and (b) when to release the results of computation (more precisely, when to release a mined block). By considering only complete information games, we may weaken the immediate applicability of the results, but we obtain a clean framework for studying these issues with proper focus and in a rigorous way.

We consider two stochastic games: (a) the immediate-release game in which every miner releases immediately the blocks that he mines; in this game, the strategy of every miner is to select an appropriate block to mine and (b) the strategic-release game in which miners not only can select which block to mine, but they can also withhold releasing blocks; we add the interesting twist that they must immediately announce a successful mining of a block but not the block itself; in this way, all miners have complete information of the current situation but can mine only the blocks which have been completely released by their discoverers. This interesting twist turns a very complicated game of incomplete information into an attractive, highly non-trivial game with complete information. We believe that although this is not the game actually modeling the real world\footnote{It is interesting to point out that there are potential real world settings where the fact that a block was mined becomes known but the block itself may not be released. This would correspond to a setting where a head of a pool withholds a block but a pool participant that found the block announces its discovery. }, it is of great value in understanding the game-theoretic aspects of the highly convoluted incomplete information game that takes place in the real world.

\subsection{The Bitcoin currency and the blockchain protocol}
\label{sec:bitc-blockch-prot}

The Bitcoin currency came into existence in 2009 when the first $\bitcoin 50$ (bitcoins) were introduced to the system. Since then it has gained significant recognition with more than $\bitcoin 15$ million  currently in circulation with an exchange rate of $\bitcoin 1$ for around $500$\euro\ at the time of this writing. 

Bitcoin was the first successful design of a fully distributed currency and is based on a peer-to-peer network achieving consensus on the broadcasted bitcoin transactions. The difficulty in the design of a fully distributed digital currency lies in avoiding the possibility of double spending the (easily reproducible) digital coins,  \cite{Nakamoto08,Finney:2011:DoubleSpendingAttack,Karame:2012:DoubleSpending}. Eliminating these possibilities is far from trivial in a decentralized setting, with applications that go way beyond the monetary nature of bitcoin. A practical pure consensus protocol can have far-reaching implications in many domains, e.g., secure distributed timestamping \cite{Haber:1991:Timestamping} or decentralizing Internet name services,
(see also \cite{Pease:1980:Agreement} \cite{Okun2007:ByzantineAgreement} regarding the topic of Byzantine Agreement in the standard and anonymous version respectively). 

Bitcoin's design is based on several previous attempts. The idea is that the users maintain a public log of all transactions that have taken place between the clients (bitcoin owners). The history of the recorded transactions alone determines the ownership of the bitcoins, so it is imperative that the users reach an agreement about it. Perhaps the most important element of the Bicoin protocol, allowing it to provide a consensus solution in a setting where identities are not available, is the notion of \emph{proof of work}  \cite{Dwork1993PricingProcessing}  \cite{Back-Hashcash} \cite{JuelBrainard-cryptographic99} \cite{RivestShamirWagner-TimeLock}. Bitcoin uses cryptography to make the miners provide proof of work before validating a block of transactions in an operation that is reminiscent of Sybil attack prevention techniques that were discussed earlier, see e.g., \cite{ajk2005}. Cryptographic techniques (digital signatures) are also used to guarantee that only the rightful owner of bitcoins can ``spend'' them by including them in a transaction.

We now give a brief but detailed description of how the Bitcoin protocol works. As mentioned before, the goal is for the users to maintain a public ledger listing all transactions of the form \emph{\{X pays Y the amount of \bitcoin Z\}} in a distributed fashion, i.e., without the need for a central authority. This is called the \emph{blockchain}. To add a block of transactions to the blockchain and claim the corresponding reward, a user has to solve a hard cryptographic puzzle, thus spending computational power. This process is called mining and we will use the term \emph{miner} to refer to the users.

Ideally, the blockchain would be a simple chain of blocks implying precedence between the corresponding transactions, i.e., a serialization of valid transactions between the clients of the Bitcoin protocol. This would be the case if miners always started mining at the last announced block and propagated each block creation immediately to the network of the remaining miners. However, the selfish nature of the miners who try to receive the rewards of as many blocks as possible (or even the inherent delay of block propagation in the distributed network)\footnote{Accidental forking happens every 60 blocks on average.} can result in temporary forks in the blockchain. The protocol suggests to the miners to always start mining at the end of the branch which needed the largest amount of computational effort so far, i.e., the end of the longest fork. This strategy is called \frontier\ and we will call the miners that follow it \emph{honest} miners. 

In order to mine a block, the miners are asked to compute some nonce value such that adding that value to the hash of the preceding block (specifying the branch of the blockchain that they are extending) and the batch of transactions they are trying to validate will render a block whose hash (SHA-256) does not exceed a certain threshold. Solving this puzzle is assumed to be computationally hard; however, the verification of a provided solution is easy. It is also assumed that in order to solve the puzzle, a miner cannot do any better than trying it for different inputs repeatedly. If a miner has more computational power than another, then he has higher probability to solve the puzzle faster and create the block. The difficulty of the puzzles is adjusted so that a single block is created every $10$ minutes on average. Once a miner solves a puzzle and creates a block, he becomes eligible to receive a reward. It is important for our game-theoretic analysis that the reward is given only if the corresponding block is permanently added to the blockchain. The reward for each successful mining is a fixed amount of newly created bitcoins plus fees from the transactions that are included in the block. This fixed reward was originally set to $\bitcoin 50$  while the protocol determines that it will be halved each time $210,000$ blocks are permanently added to the blockchain, which happens approximately every $4$ years. This implies that the protocol will eventually stop creating new bitcoins ($10^{-8}$ BTC is the minimal unit of Bitcoin), at which time the fixed rewards for mining will be entirely replaced by transaction fees.\footnote{Currently, the transaction fees are insignificant compared to fixed rewards (approximately 1.33\%).
}

The reward structure of the protocol guarantees that the honest miners' revenue is proportional to their computational power. However, understanding when it is profitable for the miners to deviate from the honest strategy is a central question and has attracted a lot of attention. 
The original assumption was that no miner has an incentive to deviate from the honest strategy if the majority of the miners are honest. However, this is not true as was shown by Eyal and Sirer \cite{Eyal2014MajorityIsNotEnough}. They gave a specific strategy which, when followed by a miner with computational power at least $33\%$ of the total power, provides rewards strictly better than the honest strategy (assuming that every other miner is honest). This was extended computationally in
 \cite{SapirsteinSompolinskyZohar-OptimalStrategies}.

\subsection{Our results}
\label{sec:our-results}

We consider two stochastic games whose states are rooted trees. The nodes of the tree are blocks that have been mined in the past. At every time-step, each miner selects a node of the current tree and tries to extend it by one new block. The probability that a miner succeeds in mining a new block is proportional to the miner's computational power. The utility of a miner is the fraction of successfully mined blocks in the common history, i.e., in the longest path from the root.

In the \emph{immediate-release} game, the miner can select any node of the tree to mine, while in the \emph{strategic-release}, the miner can select only nodes that have been (declared) \emph{released} by their creator.

If a miner has computational power above a threshold, he may not mine a node at the frontier (the set of deepest nodes) in the hope that his blocks will become the accepted history instead of the already mined blocks. Also, he may not release a node immediately in the hope that the computational power of the other miners will be wasted in mining blocks which will not be part of the common history.

Our work seeks to identify the thresholds of the computational power below of which \frontier\ (the honest strategy) is a Nash equilibrium while above them it is no longer the optimal strategy. We denote by $h_0$ and $\hat h_0$ the threshold for the immediate-release and the strategic-release games, respectively. It is easy to see that for both games $h_0, \hat h_0 \in [0,0.5]$. Our results are as follows. \begin{itemize} \item We prove that in the immediate-release setting the threshold is $0.361\leq h_0\leq 0.455$. This implies that a miner with at most $36\%$ of the total computational power can not gain more than $36\%$ of the total rewards, i.e., his fair share which he would gain by being honest and following the suggested strategy \frontier\mbox{}; and that a miner with computational power more than $46\%$ will always deviate from the honest strategy. An immediate consequence is that if every miner has computational power less than $h_0$, then \frontier\ is a \emph{Nash equilibrium}.  We have experimentally determined that the actual threshold is close to $h_0\approx 0.42$.

\item Regarding the strategic-release setting, we prove rigorously that the corresponding threshold is lower bounded by $\hat h_0\geq 0.308$ (root of the polynomial $p^3- 6 p^2+ 5 p - 1$). A similar result was obtained recently by Sapirshtein, Sompolinsky, and Zohar in \cite{SapirsteinSompolinskyZohar-OptimalStrategies}. Their result is tighter and sets the threshold below $0.33$ (cf. Figure~2 in their paper for $\gamma=0$). The difference between our approach and the approach of \cite{SapirsteinSompolinskyZohar-OptimalStrategies}, is that they provide an algorithm to compute an approximately optimal strategy and then run their algorithm to estimate the threshold. The authors acknowledge the fact that their ``algorithm copes with computational limitations by using finite MDPs as bounds to the original problem, and by analyzing the potential error that is due to inexact solutions''. Our result on the other hand, gives an exact, albeit suboptimal bound on the threshold and considers exact best responses in a purely mathematical (not computational) way.
\end{itemize}

\paragraph{Open problems:} This work sets the stage and makes progress towards a systematic study of these complicated stochastic games. A lot of issues remain open. For example, besides the obvious problem of tightening our results, there are a lot of interesting questions about the Nash equilibria above the thresholds $h_0$ and $\hat h_0$ in both types of games.

\subsection{Related work}
\label{sec:related-work}

The Bitcoin protocol was originally introduced in \cite{Nakamoto08} and was built based on ideas from \cite{Back-Hashcash} and \cite{Dai:1998:BMoney}. After the ``creation'' of Bitcoin several other alternative electronic currencies followed, known as \emph{altcoins} (http://altcoins.com/), e.g., litecoin, Primecoin etc. The Bitcoin white paper provides a probabilistic analysis of double spending attacks 
while a more detailed analysis can be found in \cite{Rosenfeld-Doublespending}. Bitcoin's design and research challenges are discussed in \cite{BonneauMillerEtal-SoKsurvey15} along with a presentation of the existing research. In \cite{tschorsch:2015:BitcoinAndBeyond} an extensive and more introductory survey on distributed cryptocurrencies can be found.


The works most relevant to ours are  \cite{KrollDaveyFelten:2013:Economics}, \cite{Eyal2014MajorityIsNotEnough} and \cite{SapirsteinSompolinskyZohar-OptimalStrategies}.
In \cite{KrollDaveyFelten:2013:Economics}, the equilibria of the Bitcoin game are considered. The authors observe that any \emph{monotonic} strategy is a Nash equilibrium (one of many).
Their analysis though is quite restricted: as we show, even in the case of the immediate release
game, \frontier\ is not best response for values above $0.455$.
%
In  \cite{Eyal2014MajorityIsNotEnough}, the authors prove that a guaranteed majority of honest miners is not enough to guarantee the security of the Bitcoin protocol. In particular they present a specific strategy called the ``Selfish Mine'' strategy and examine when this strategy is beneficial for a pool of miners. It appears that a fraction of   $1/3$ of the total processing power is always enough for a pool of miners to benefit by applying the Selfish Mine Strategy no matter the block propagation characteristics of the network ($\gamma=0$ in their setting). Hence, this constitutes a profitable attack against the Bitcoin protocol.
Furthermore, when   block propagation is favoring the attacker the threshold above which ``Selfish Mine'' is beneficial is any non-zero value.
It is not hard  to see that the ``Selfish Mine'' strategy does not fully exploit settings
where block propagation is more favorable to the attacker. This was observed in
\cite{GarayKiayiasLeonardos-Backbone} where an optimal attack against the property of
chain quality was considered in their setting where the network is considered
to be completely adversarial.
In \cite{SapirsteinSompolinskyZohar-OptimalStrategies} the authors consider a wider set
of possible strategies that includes the ``Selfish-Mine'' strategy and explore
this space computationally.
Their analysis also accounts for possible communication delays in the network, the presence of which can diminish the profit threshold.

A vast majority of previous work examines possible types of attacks against the Bitcoin protocol and suggest adaptations of the protocol to ensure its security. We very briefly mention some of these works here.

Successful pool mining related attacks are discussed in \cite{Rosenfeld-RewardSystems} and \cite{CourtoisBahack:2014:BlockWithholding}. In \cite{Eyal:2014:Dilemma}
the author considers attacks performed between different pools where users are sent to infiltrate a competitive pool giving raise to a \emph{pool game}. See also \cite{Lewenberg2015GameTheoretic} for a (cooperative) game theoretic analysis regarding pool mining.
\cite{Babaioff:2012:RedBalloons} deals with information propagation and Sybil attacks. 
The authors propose a reward scheme which will make it in the best interest of a miner to propagate the transactions he is made aware of and not duplicate. 
\cite{KrollDaveyFelten:2013:Economics}
considers an attack that can be performed from people that are only interested in destroying Bitcoin, as opposed to other attacks performed by users trying to increase their expected reward. This is called the \emph{Goldfinger} attack.
\cite{Heilman:2015:EclipseAttacks} focuses on the peer-to-peer network and examines \emph{eclipse} attacks where the attacker(s) isolates a node/user from the network 
and forces him to waste his computational power thus participating in an attack without even being aware. 
Certain deanonymization attacks have also recently been observed \cite{Meiklejohn:2013NoNames} by analysing the transactional graph (see also \cite{ReidHarrigan:2012:Anonymity} and \cite{Ron:2013:TransactionGraph}).


In \cite{SompolinskyZohar-SecureHighRate} an alternative consensus method is described, called Greedy Heaviest-Observed Sub-Tree or GHOST. A variant of GHOST has been adopted by Ethereum, a distributed applications platform that is built on top of block chains. 
\cite{EyalEtal-bitcoinNG} attempts to overcome scalability issues that arise in Bitcoin (block size and interval vs latency and stability) by proposing a new scalable blockchain protocol. 
\cite{PoonDryja-Lightning15} provides another alternative for scalability that utilizes off chain transactions while using the distributed ledger for maintaining the contracts between the parties that engage in off chain transactions. 
In \cite{GarayKiayiasLeonardos-Backbone} the authors analyze the Bitcoin protocol in depth. They abstract the core of the protocol that they term the {\em bitcoin backbone} and prove formally its main attributes and properties, which can be used as building blocks for achieving goals other than simply maintaining a public ledger. 
They show that when the propagation delay in the network is relatively small, an honest majority of users is enough to guarantee smooth operation in a cryptographic model where the rationality of the players is not considered.


%% file: model.tex
\section{The Bitcoin Mining Game and its
  Variants} \label{sec:bitcoin-mining-game}

The game-theoretic issues of bitcoin mining can be captured by the
following game-theoretic abstraction. The parameters of the game are:
\begin{itemize}
\item the number $n$ of miners or players
	
\item the probabilities $p=(p_1,\ldots,p_n)$ that miners succeed in
  solving the crypto-puzzle; these are proportional to their
  computational power and they sum up to 1: $\sum_{i=1}^n p_i=1$.

\item the depth of the game $d$; the payment for mining a new block is
  not paid immediately, but only after a chain of certain number of
  new blocks is attached to it; in the current implementation of the
  Bitcoin protocol this number is $d=100$. We will mainly consider
  games with $d=\infty$, but we will discuss briefly how they are
  affected by this parameter.
\end{itemize}

Two more parameters could play a role in a more general model of the
protocol.  The computational cost $c^*$ of mining a new block and the
reward (payment) $r^*$ for it (currently
approximately 25 bitcoins). Here we assume that the reward $r^{*}$ is
constant and we scale all payments so that $r^{*}=1$. Also, if the
expected gain is high enough to entice a miner to
participate,\footnote{ The protocol must satisfy Individual
  Rationality. Rational participating players should have non-negative
  expected utility.} its actual value is not important, since the
miner tries to maximize revenue.

Note that because of the distributed nature of the
  Bitcoin protocol, it is possible that more than one miners succeed
  almost simultaneously to mine a new block. We choose to ignore this
  aspect here; nevertheless, it is easy to generalize our model to a setting
where $\sum^n_{i=1} p_i < 1$ and thus there is a non-zero probability at each
step that no miner will be awarded a block.

During the execution of the protocol, the miners build a tree of nodes/blocks to which they try to add more blocks. The protocol aims to
increase the height of this tree by one every time a new node is created (every ten minutes on average) but this is not necessarily consistent with the incentives of the players who might chose to mine blocks that are not the deepest ones of the tree. Once a
miner succeeds in creating a block, the new node is supposed to be added to the tree. However if the
miner is strategic, he may have reasons not to add the newly
discovered node to the tree. Therefore besides the publicly known
tree, each miner might have his own private tree.

\begin{definition}[State]
  A public state is simply a rooted tree. Every node is labeled by
  one of the players. The nodes represent mined blocks and the label
  indicates the player who mined the block. Every level of the tree
  has at most one node labeled $i$ because there is no reason for a
  player to mine twice the same level.

  A private state of player $i$ is similar to the public state except it may contain more nodes called private nodes and labeled by
  $i$. The public tree is a subtree of the private tree and has the
  same root.
\end{definition}

In the incomplete information case, the private states may also
include the partial knowledge that players have about the other
players (knowledge about the probabilities of other private trees, but
also about their knowledge etc). This is a very complicated case, and
we do not treat it in this work. Instead we treat two complete-information cases in which the private states of all miners are common knowledge:
\begin{description}
\item[Immediate-release model] Whenever a miner succeeds in mining a
  block, he releases it immediately, and all miners can continue
  from the newly mined block.
\item[Strategic-release model] Whenever a miner succeeds in mining a
  block, it becomes common knowledge. However, the miner may decide to
  postpone the release of the block. Until the block is released,
  other miners cannot continue mining from this block, although they
  are aware of its existence.
\end{description}
While the second model has no counterpart in practice,
we believe it is of high theoretical interest as it can serve
as an intermediate model between immediate-release and strategic-release
with incomplete information. The immediate-release model enables the study
of miners that follow the protocol in terms of block propagation but
mine strategically, while the strategic release allows us to extend
the study to the game-theoretic issues of block witholding.
Although ideally we would want to study the latter
under the incomplete information regime, we will defer this for future
work since the game becomes substantially more complex to analyze.
It is important to stress that any strategy in the
full information setting is also a valid strategy in the
incomplete information setting. Importantly, if a strategy is not
dominant in the full information setting it cannot be dominant
in the incomplete information setting.

We describe the set of strategies for the strategic-release case. The
immediate-release case is the special case in which the release
function has been fixed to immediate-release.
\begin{definition}[Strategy]
  A pure strategy of player $i$ consists of two functions ($\mu_i$,
  $\rho_i$):
  \begin{itemize}
  \item the mining function $\mu_i$ which selects a node of the
    current public state to mine.
  \item the release function $\rho_i$ which is a (perhaps empty)
    private part of the player's state which is added to the public
    state.

    Both functions depend on the state of knowledge of the miner.
  \end{itemize}
\end{definition}
The suggested strategy by the designer of the protocol is the
\frontier\ strategy.
\begin{definition}[\frontier]
  We say that a miner follows the \frontier\ strategy when he releases
  any mined block immediately and selects to mine one of the
  deepest nodes.
\end{definition}
In the expected execution of the
protocol in which all players play the \frontier\ strategy, the
Bitcoin protocol creates a path. This is due to our simplifying
assumption that no players simultaneously mine a block. In practice,
the \frontier\ strategy creates something very close to a path,
with occasional ``orphan'' blocks hanging from it.

The game is played in phases. In each phase, each player $i$ uses his
mining function $\mu_i$ to select a block to mine. We assume that
exactly one player succeeds in mining a block in each
phase
, and that
the probability of success for each player is given by the
probabilities $(p_1,\ldots,p_n)$. The winner then adds the newly mined
block to his private tree as a leaf hanging from a previously existing node. He
then applies his release function which may add some of his private
part of the tree (for example, the newly mined node) to the public
tree. This may trigger a cascade of releases from other
players. When the dust settles, we will have a new public tree, and
each player will have updated knowledge\footnote{The releasing step
  is non-deterministic and, depending on the release functions and the
  order of applying the release functions, may lead to different outcomes. However, this never happens in the cases we analyze
  here.}. The phase ends at this point and a new phase begins.

Note that it is possible that the release function of the winning
player may result in an empty release. Since here we consider the
complete information case, all miners can immediately detect the end
of the phase. In the incomplete information case however this is not
possible, although the miners can estimate the probability of this
happening, and that adds another complication in modeling strategic
considerations in the incomplete information regime.

Payments for mining new blocks are essential to incentivize the
players to try to mine new blocks. A miner who succeeds in mining a
block is paid $d$ phases later (currently $d=100$); the delay is
considered sufficient to guarantee that no long branches off the main
path exist. The description of the payment scheme seems sufficient
\emph{under the assumption that branches become stale quickly and that only
  the main trunk survives}. With the term \emph{trunk} we refer to a long
path with ignored stale branches (i.e., the sibling of a node that is paid for as
well as its descendants will get no payment at any point in time, thus
they are effectively deleted).

A rigorous game-theoretic analysis of the Bitcoin protocol is quite
 complicated because of the potential strategic branching, and it
requires a more precise definition of payments. To be consistent with
the non-game-theoretic considerations of the Bitcoin protocol, we
assume that at every level (i.e., height of the tree) only one node is
paid for, the first one which succeeds in having a descendant $d$
generations later. In graph-theoretic terms, a node $u$ is paid for when
its path from the root is extended by a path of length $d$; when this
happens every sibling (as well as its descendants) of node $u$ becomes
stale.

\begin{definition}[Payments]
  For some nodes of the tree, the miners who discovered them will get
  a fixed payment (normalized to 1). The payments comply with the following rules:
  \begin{itemize}
  \item the nodes that receive payment must form a path from the
    root. This immediately adds the restriction that at every level of
    the tree exactly one node receives payment.
  \item among the nodes of a single level that satisfy the above path
    restriction, the first one which succeeds in having a descendant
    $d$ generations later receives payment.
  \end{itemize}
  Since only one node per level is paid for, the utility of a miner in
  the long run is defined as the fraction of the total payment which
  he receives (his paid nodes over the total number of paid nodes).
\end{definition}

When a node is paid for, rational miners will completely ignore every
branch that starts at an earlier node. So in the long run, the tree
essentially becomes the trunk with a small tree of depth at most $d$ at the end. We will
call such a game \emph{truncated at level $d$}. Immediate-release
truncated games are \emph{finite stochastic games}.

We will also consider games that are not truncated at a specific level
$d$. We have to do this with care, since it is possible that two or
more miners will continue expanding their own branch forever and they
will never agree. However, in our games this cannot happen in an
optimal play when one miner has probability less than $1/2$. The
reason is that the utility of a miner is the fraction of the total
payment he receives which is expected to be 0 if he keeps mining his
own branch forever (a case of gambler's ruin).


%% file: immediate_release.tex
\section{The Immediate-Release Game}
In this section we determine the conditions which guarantee that the suggested \frontier\ strategy is a Nash equilibrium. We fix the strategy of all but one miners to \frontier\ and identify when \frontier\ is the best response of the remaining miner.

We can assume that all miners who follow the \frontier\ strategy by
assumption act as a single miner (as mentioned, we consider only the simplified setting where
miners do not simultaneously produce blocks). This gives rise to a two-player
(two-miner) game: Miner 1 is the miner whose optimal strategy (best
response) we wish to determine and has relative computational power
$p$ ($p$ fraction of the total computational power), while Miner 2 is assumed to follow the \frontier\ strategy and have collective relative computational power $1-p$.

A (public) state is simply a rooted tree of width at most 2. In the
immediate-release case, after pruning away stale (abandoned)
branches, the state is a long path (called trunk) followed by two
branches, one for each miner of lengths $a$ and $b$ (see
Figure~\ref{fig:states}). The lengths of these two branches determine
the state and can be $0$. Also because Miner 2 plays \frontier, his
path must be the longest one, except temporarily when Miner 1 mines a
block and moves ahead; in this case we have $a=b+1$, and when this
happens, Miner 2 abandons his path and continues from the frontier of
the other path. To summarize, the states of the game are the pairs
$(a,b)$ with $0\leq a\leq b+1$.

\begin{figure}
  \centering
  \includegraphics[scale=0.3]{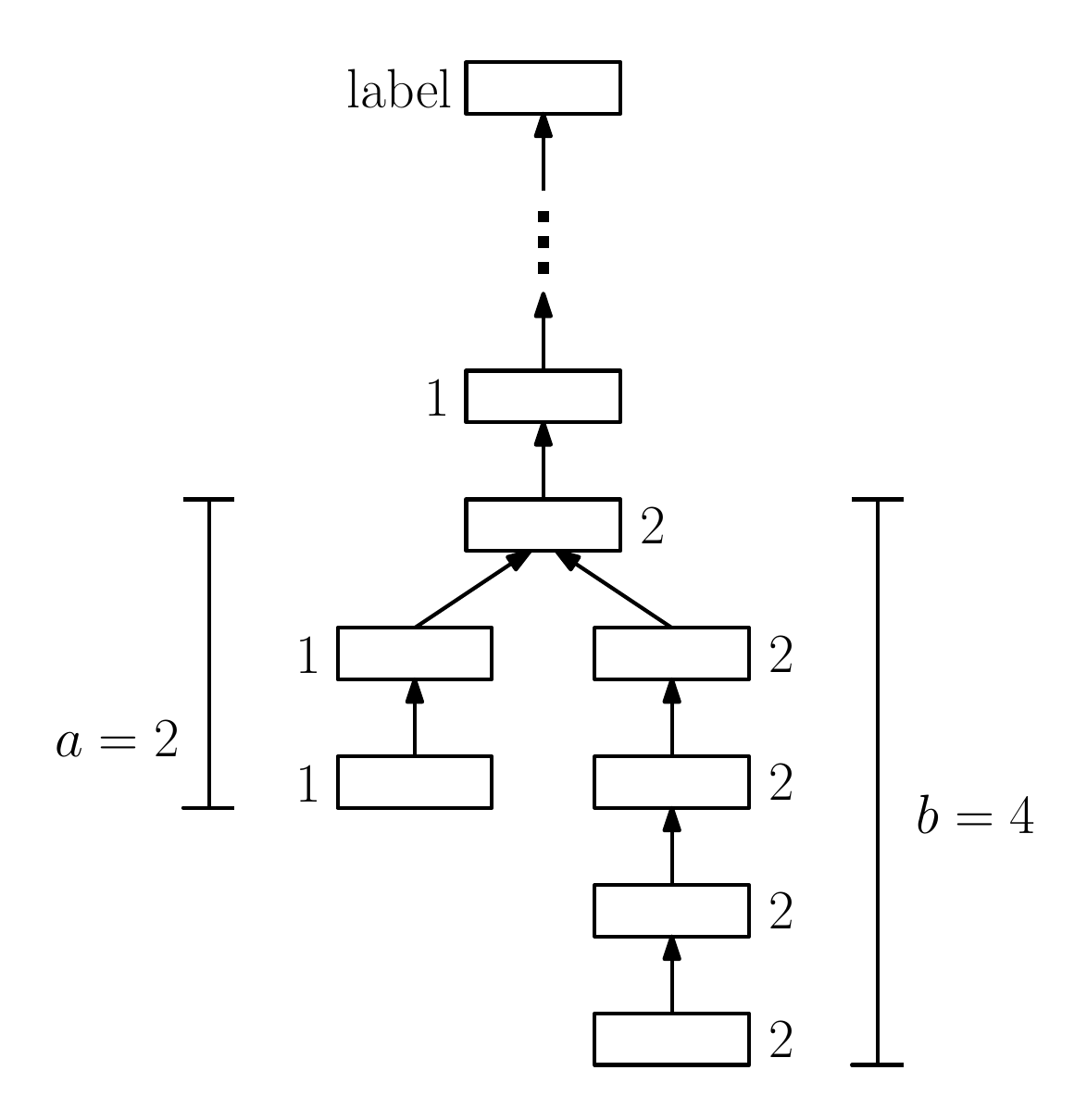}
  \hspace*{4cm}
  \includegraphics[scale=0.5]{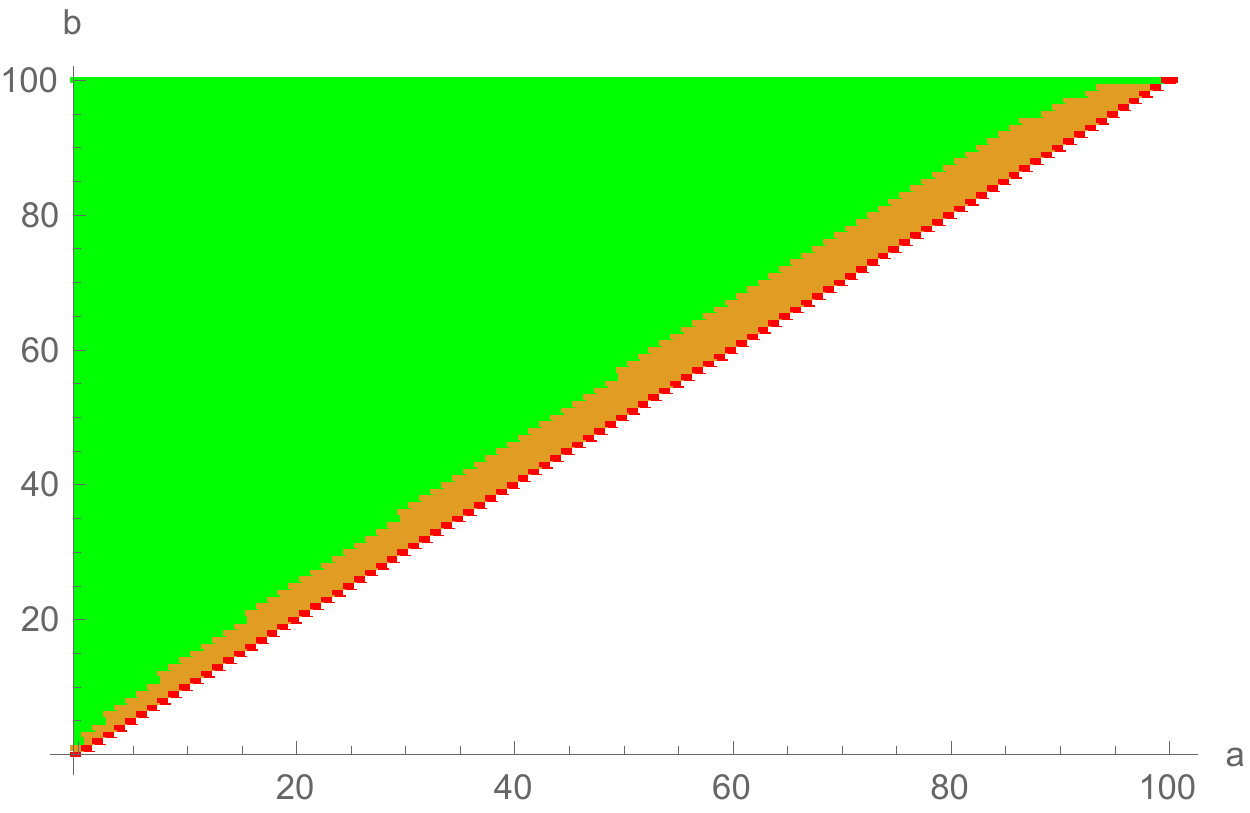}
  \caption{One the left, a typical state (tree). On the right, the set
    of states form the optimal strategy for the truncated game at
    $d=100$. The upper-left green part is the set $C$ of capitulation
    states, the diagonal red line is the set $W$ of winning states,
    and the orange part the set $M$ of mining states.}
  \label{fig:states}
\end{figure}

The set of states $(a,b)$ can be partitioned into three parts (see
Figure~\ref{fig:states}):
\begin{description}
\item[Mining states:] the set $M$ in which both miners keep
  mining their own branch. State $(0,0)$ belongs to $M$. Depending on the strategy of miner 1, additional states
  might be included in $M$.
\item[Capitulation states:] the set $C$ of states in which Miner 1
  gives up on his branch and continues mining from some block of the
  other branch. When the game is truncated at depth $d$, this set
  includes all states of the form $(a,d)$, for $a=0,\ldots,d$.
\item[Winning states:] the set $W$ of states in which Miner 2
  capitulates. Given that miner 2 plays \frontier\ it holds
  that $W=\{(a, a-1)\,:\, a\geq 1 \}$.
\end{description}

Note that when Miner 1 capitulates and abandons his own branch, he
can choose to move to any state $(0,s)$. Since the miner is rational,
he will select the best of these states (we assume that in case of a
tie, he will always select the same state). It follows that without loss of generality we can describe
fully the strategy of Miner 1 by specifying the set of mining states
$M$ and the state $(0,s)$ that is chosen when Miner 1 capitulates. The \emph{set of
  deterministic strategies} of Miner 1 is exactly the set of pairs
($M$, $s$), where $M$ is the set of mining states and $(0,s)$ is the
landing state to which the miner jumps from any capitulation state.

Let $g_k(a,b)$ denote the expected gain of Miner 1 when the branch of the honest miner in the execution tree is extended by $k$ new
levels starting from an initial tree in which the
branches of Miner 1 and 2 have lengths $a$ and $b$, respectively. It
should be intuitively clear that, in the long run, the expected gain
per level should be almost independent of the initial state. As a result, for large $k,k'$, it should be that $g_k(a,b) - g_{k'}(a,b)$
is independent of $a,b$ and thus it is merely a function of $k,k'$.
Furthermore, assuming that  constant rewards are given
and expenses per block created remain constant, it will hold that $g_{k}(a,b)$
satisfies for any large $k,k'$,
$$\frac{g_k(a,b) - g_{k'}(a,b)}{k-k'} = g^*.$$
for some constant $g^*$ which represents the \emph{expected gain per level}
in the long run.  Based on this we can define the expected gain as
$$g_k(a,b)=k\cdot g^* + \f(a,b),$$
where the potential function
$\f(a,b)=\lim_{k\rightarrow \infty} g_k(a,b) - k\cdot g^*$ denotes the
advantage of Miner 1 for currently being in state $(a,b)$; that this
limit exists follows from a straightforward argument\footnote{Since
  state $(0,0)$ is clearly recurrent, we can alternatively define
  $\f(a,b)$ as the expected value $g_k(a,b) - k\cdot g^*$ until state
  $(0,0)$ is reached.}. The objective of Miner 1 is to maximize
$\g$.

For a strategy $(M,s)$, we can define $g_k(a,b)$ recursively as
follows. When the current state is $(a,b)$, there are three
possibilities:
\begin{itemize}
\item If $(a,b)\in M$, both miners continue mining. With probability
  $p$, Miner 1 succeeds to mine the next block first and the new
  state is $(a+1,b)$; with the remaining probability Miner 2
  succeeds and the new state is $(a,b+1)$.
\item If $(a,b)\in C$, Miner 1 abandons his branch and the new state
  is $(0,s)$.  The trunk increases by $b-s$ blocks.
\item If $(a,b)\in W$, Miner 2 abandons his branch and the new state
  is $(0,0)$. Miner 1 receives payment for $a$ blocks. The trunk increases
  by $a$ blocks.
\end{itemize}
The frontier advances when either Miner 2 wins or when Miner 1 wins
and $a=b$, thus reaching a winning state. From the above consideration, we get
\begin{equation*}
    g_k(a,b)=
    \begin{cases}
       g_{k-1}(0, 0) + a  & a=b+1 \\
       \max\left( \max_{s=0,\ldots,b-1} g_k(0,s), p g_k(a+1, b) + (1-p) g_{k-1}(a,
       b+1) \right) & \text{otherwise}
    \end{cases}
\end{equation*}
and by definition $g_0(a,b)=0$. From this, we can get a similar recurrence for $\f$:
\begin{eqnarray*}
    &&\f(a,b)=\begin{cases}
       \f(0, 0) + a-\g  & a=b+1 \\
       \max\left( \max_{s=0,\ldots,b-1} \f(0,s), p \f(a+1, b) + (1-p) \f(a,
       b+1) - (1-p) \g\right) & \text{otherwise}
    \end{cases}
\end{eqnarray*}We also fix $\f(0,0)=0$; note that the potential of all states is non-negative.

We note that the above definitions do not take depth $d$ into account (consider $d=\infty$). This is without loss of generality for the proof of Theorem \ref{thm:lower-bound} where we use this recurrence, as proving that \frontier\ is best response for $d=\infty$ (i.e., when Miner 1 has a superset of available winning paths than for any constant $d$) implies that the result also holds for any constant $d$.



\subsection{\frontier\ is a NE iff \boldmath$p\leq h_0$, where $0.361\leq h_0\leq 0.455$}
If every miner plays \frontier, their expected gain is proportional to
the probability of mining a block and therefore proportional to their
relative computational power $p_i$. From this we get the following
proposition.
\begin{proposition}\label{prop:frontierNE}
  \frontier\ is a Nash equilibrium if, having fixed the strategy of
  all miners except $i$ to \frontier, the best response of a miner $i$
  has expected gain per level equal to $p_i$.
\end{proposition}

In this section we bound the threshold on the computational power of each miner so that \frontier\ is a Nash equilibrium. The main result is the lower bound that follows.
\begin{theorem} \label{thm:lower-bound} In the immediate-release
  model, \frontier\ is a Nash equilibrium when every miner $i$ has
  relative computational power
  $p_i\leq \frac{1}{3} (1 - 8/(1 + 3 \sqrt{57})^{1/3} + (1 + 3
  \sqrt{57})^{1/3})\approx 0.361$ (root of the polynomial $2p^2-(1-p)^3$).
\end{theorem}

Starting in a mining state, one of the two miners will eventually
capitulate and join the other branch.\footnote{This happens even when
  $d$ is unbounded. For a miner with relative computational power
  $p_i<1/2$ who never capitulates, the expected gain is 0, as he is
  engaged in a gambler's ruin
  situation. Lemma~\ref{lem:win-probability} shows that the
  probability of winning drops exponentially with the distance of his
  branch from the frontier.} The probability that the miner wins such
a race plays a significant role in our analysis. A formal definition follows.

\begin{definition}
  Let $r_{M,s}(a,b)$ denote the winning probability starting at state
  $(a,b)$, that is, the probability that a winning state will be
  reached before a capitulation state.

  We will use the notation of $r(a,b)$ for the optimal strategy
  ($M$, $s$), and $r_{\infty}(a,b)$ when the miner never capitulates
  (i.e., when $C=\emptyset$).
\end{definition}

The probability $r_{M,s}(a,b)$ can be defined recursively as follows:
\begin{equation*}
  r_{M,s}(a,b)=
  \begin{cases}
    p\cdot r_{M,s}(a+1,b) + (1-p)\cdot r_{M,s}(a, b+1), & (a,b)\in M, \\
    1, & (a,b)\in W, \\
    0, & (a,b)\in C.
  \end{cases}
\end{equation*}

The next simple lemma which bounds the probability $r(a,b)$, plays a
central role in our analysis:
\begin{lemma}\label{lem:win-probability}
The following holds for every state $(a, b)$:
\begin{equation*}
  r(a,b)\leq r_{\infty}(a,b)  = \left(\frac{p}{1-p}\right)^{1+b-a}.
\end{equation*}
\end{lemma}

\begin{proof}
  The inequality follows from the fact that the miner has more
  opportunities/paths to reach a winning state when he never gives up.

  Regarding the equality, we essentially want to solve the recurrence
  for $r_{\infty}$
  \begin{align*}
  r_{\infty}(a, b) &=
     \begin{cases}
       1 & \text{when $a=b+1$} \\
       p\cdot r_{\infty}(a+1, b)+(1-p)\cdot r_{\infty}(a,
       b+1) & \text{otherwise.}
     \end{cases}
  \end{align*}
  Consider the quantity $l=1+b-a$ which captures the distance of state
  $(a,b)$ from the set of winning states. This quantity decreases by
  one when the miner succeeds in mining a block, and increases by one
  when the opponent succeeds. Since we consider $r_{\infty}$, the case
  in which the miner never gives up, the situation is a gambler's ruin
  version: the quantity $l$ behaves as the position of a biased random
  walk on a half-line with an absorbing state at 0, probability $p$ of
  moving away from the absorbing state, and probability $1-p$ of
  moving towards the absorbing state. The probability
  $r_{\infty}(a,b)$ is the probability of reaching the absorbing state
  starting at position $l=1+b-a$ which can be easily computed to be
  $(\frac{p}{1-p})^l$.
\end{proof}

The next lemma gives a very simple necessary and sufficient condition for the potential function so that \frontier\ is a Nash equilibrium. Intuitively, the condition states that Miner 1 capitulates from state $(0,1)$, so no other state except $(0,0)$ will ever be reached.

\begin{lemma} \label{thm:BR-and-potential}
  Strategy \frontier\ is a best response for Miner 1 if and only if $\f(0,1)=\f(0,0)$.
\end{lemma}
\begin{proof}
  From the definition of $\f$, we have:
  \begin{align*}
    \f(0,0) &= p\,\f(1,0)+(1-p)\,\f(0,1) -(1-p)\,\g \\
    \f(0,1) &= \f(0,0)+1-\g.
  \end{align*}
  It follows that $\f(0,1)-\f(0,0)=(\g-p)/(1-p)$. Since \frontier\
  is best response if and only if $\g=p$ (recall Proposition \ref{prop:frontierNE}), the lemma follows.
\end{proof}

We now derive a very useful relation between the expected optimal gain of a pair of states and the winning probability of one of them.
\begin{lemma}
  For every state $(a,b)$ and every nonnegative integers $c$ and $k$:
  \begin{equation*}
    g_k(a+c, b+c)-g_k(a,b)\leq c\cdot r(a+c,b+c).
  \end{equation*}
  Conversely for $p<1/2$, there is $\epsilon_{a,b}(k)$ which tends to 0 as
  $k$ tends to infinity, such that
  \begin{equation*}
    g_k(a+c, b+c)-g_k(a,b)\geq c\cdot \left(r(a,b)-\epsilon_{a,b}(k)\right).
  \end{equation*}
\end{lemma}
\begin{proof}
  We focus on the first inequality since the proof of the second
  inequality follows from similar reasoning. Suppose that the current
  state is $(a,b)$ and Miner 1 continues playing not in the optimal
  way, but by simulating the strategy that he would follow had the
  current state been $(a+c,b+c)$. The crucial observation is that the
  simulation can be carried out because the strategy of the other
  miner, based only on the difference $b-a$, is unaffected.

  Let $\bar g_k(a,b)$ be the gain for the next $k$ levels using this
  potentially suboptimal strategy. Let also $\bar r$ be the
  probability that Miner 1 will reach a winning before a capitulation
  state (within the next $k$ levels). Then we must have
  \begin{equation*}
    \bar g_k(a,b) = g_k(a+c,b+c) - c\cdot \bar r.
  \end{equation*}
  Now, since the simulated strategy cannot be better that the optimal
  strategy, it is clear that $g_k(a,b)\geq \bar g_k(a,b)$. Furthermore, the probability $r(a+c,b+c)$, which is the
  probability that Miner 1 wins the race even if more than $k$
  levels are used, is at least equal to $\bar r$. These two bounds
  give the first part of the lemma:
  \begin{equation*}
    g_k(a,b)\geq \bar g_k(a,b)= g_k(a+c,b+c) - c\cdot \bar r \geq
    g_k(a+c,b+c)- c\cdot r(a+c,b+c).
  \end{equation*}
  The second inequality follows from similar considerations and in
  particular by starting in state $(a+c,b+c)$ and simulating the
  strategy as being in state $(a,b)$. The term $\epsilon_{a,b}(k)$ is
  needed because now we want to bound the probability $\bar r$ from
  below: $\bar r = r(a,b)-\epsilon_{a,b}(k)$, where
  $\epsilon_{a,b}(k)$ is the probability that Miner 1 will win the
  race in more than $k$ levels. This probability tends to 0 as
  $k$ tends to infinity; in particular for $p<1/2$, this is bounded by
  the probability that a gambler with initial value $a$ and
  probability of success $p<1/2$ will win against a bank of initial
  value $b$ after $\Theta(k)$ steps.
\end{proof}

The second part of the previous lemma will not be used in this work,
but it may prove helpful to tighten our results. The following
corollary is a direct consequence of the first part of the lemma.

\begin{corollary}\label{cor:bound-restriction}
  For any state $(a,b)$ and nonnegative integer $c$
\begin{align*}
\f(a,b)\geq \f(a+c,b+c)-c\cdot r(a+c,b+c).
\end{align*}
\end{corollary}

It should be intuitively clear that the gain $g_k(a,b)$ is non-decreasing in $a$ as having mined more
blocks cannot hurt the miner. Therefore we get the following useful fact which can also be easily proved by double backwards induction on $b$ and $a$.
\begin{proposition}\label{lem:struct-monotonicity}
The potential $\f(a,b)$ is non-decreasing in $a$.
\end{proposition}

The following three lemmas use the above results to provide explicit bounds on the potential of states $(1,2)$, $(0,2)$, and $(0,1)$ under certain assumptions on the best response of Miner 1 and his computational power.
\begin{lemma}
  For every $p$:
  \begin{align}\label{eq:f12-upper-bound}
  \f(1,2) &\leq \frac{2p^2-p}{(1-p)^2} + \g \frac{1}{1-p}.
  \end{align}
\end{lemma}
\begin{proof}
  Consider the potential of state $(1,1)$. We know that
  \begin{equation*}
    \f(1,1) \geq p\f(2,1)+(1-p)\f(1,2)-\g(1-p) = p(2-\g)+(1-p)\f(1,2)-\g(1-p).
  \end{equation*}
  On the other hand, from Corollary
  \ref{cor:bound-restriction} and Lemma \ref{lem:win-probability} we can
  bound it from above by $\f(1,1)\leq \f(0,0)+r(1,1)\leq
  \frac{p}{1-p}$. By putting the two bounds together and eliminating
  $\f(1,1)$ we get the lemma.
\end{proof}

\begin{lemma} \label{lem:(0,2)-mining}
  For $p < (3-\sqrt{5})/2\approx 0.382$, if state $(0,2)$ is a mining
  state, i.e. $(0,2)\in M$, then
  \begin{align}
    \f(0,2) & \leq \frac{2p^2-(1-p)^3}{(1-p)^2}.
  \end{align}
  It follows that for $p< \frac{1}{3} (1 - 8/(1 + 3
  \sqrt{57})^{1/3} + (1 + 3 \sqrt{57})^{1/3})\approx 0.361$ (root of
  the polynomial $2p^2-(1-p)^3$), state
  $(0,2)$ is not a mining state.
\end{lemma}
\begin{proof}
  To bound $\f(0,2)$ we use the bound for $\f(1,2)$ from the previous
  lemma and a bound for $\f(0,3)$. To bound $\f(0,3)$ we use the
  monotonicity property of $\f(a,b)$ with respect to $a$ and apply
  Corollary \ref{cor:bound-restriction} and Lemma
  \ref{lem:win-probability} to $\f(1,3)$:
  \begin{equation*}
    \f(0,3)\leq \f(1,3) \leq \f(0,2)+r(1,3) = \f(0,2)+\frac{p^3}{(1-p)^3}.
  \end{equation*}
  Assuming now that $(0,2)$ is a mining state, we get:
  \begin{align*}
    \f(0,2) &=  p\f(1,2)+(1-p)\f(0,3) - \g(1-p)\\
    & \leq p\f(1,2)+(1-p)\f(0,2)+(1-p)\frac{p^3}{(1-p)^3} - \g(1-p)
  \end{align*}
and by solving for $\f(0,2)$:
\begin{align*}
  \f(0,2) & \leq \f(1,2) + \frac{p^2}{(1-p)^2} -\g \frac{1-p}{p} \\
  & \leq \frac{2p^2-p}{(1-p)^2} + \g \frac{1}{1-p}+  \frac{p^2}{(1-p)^2} -\g
    \frac{1-p}{p} \\
   & =  \frac{3p^2-p}{(1-p)^2} -\g \frac{(1-p)^2-p}{p(1-p)}.
\end{align*}
The coefficient of $\g$ is negative for $p\leq (3-\sqrt{5})/2$, so we
can replace it by $p$ to get the first part of the lemma (recall that
$\g$ is always greater or equal to $p$).

For the second part of the lemma, it suffices to observe that for
$p<\frac{1}{3} (1 - 8/(1 + 3 \sqrt{57})^{1/3} + (1 + 3
\sqrt{57})^{1/3})$ the expression $\frac{2p^2-(1-p)^3}{(1-p)^2}$ is
negative. Since the potential cannot be negative, it follows by
contradiction that $(0,2)\not\in M$.
\end{proof}

\begin{lemma}\label{lem:bounds-play}
  For $p < (3-\sqrt{5})/2$, if $(0,1)$ is a mining state, then $(0,2)$
  is also a mining state and
  \begin{align}
    \f(0,1) & \leq (1-p)\f(0,2)-p\frac{1-3p+p^2}{1-p}.\label{eq:bound-play-02}
  \end{align}
\end{lemma}
\begin{proof}
First we bound $\f(0,1)$ if we assume that $(0,1)\in M$. Note that this easily implies that also state $(1,1)$ is a mining state. We get
\begin{align*}
  \f(0,1) & = p\f(1,1)+(1-p)\f(0,2) - \g(1-p)\\
   & = p(p\f(2,1)+(1-p)\f(1,2) - \g(1-p))+(1-p)\f(0,2)- \g(1-p) \\
   & \leq (1-p)\f(0,2)-p\frac{1-3p+p^2}{1-p}
\end{align*}
where we used the bound for $\f(1,2)$ from above and $\f(2,1)=2-\g$.
Towards a contradiction assume now that $(0,2)\not\in M$. That is, the
miner capitulates at state $(0,2)$ and moves to either state $(0,1)$
or state $(0,0)$. Since the miner can move to state $(0,0)$ in two
steps by first moving to state $(0,1)$, we can assume without loss of
generality that he moves to state $(0,1)$ and therefore
$\f(0,2)=\f(0,1)$. By substituting this in the above expression we get
\begin{align*}
  \f(0,1)\leq - \frac{1-3p+p^2}{1-p} <0
\end{align*}
for $p< (3-\sqrt{5})/2$, a contradiction.
\end{proof}

We are now ready to prove the main result of this section. We use the bounds provided by the previous three technical lemmas (together with Lemma \ref{thm:BR-and-potential}) to prove that Miner 1 capitulates from state $(0,1)$, and that \frontier\ therefore is his best response.
\begin{proof}[Proof of Theorem \ref{thm:lower-bound}]
  For
  $p\leq \frac{1}{3} (1 - 8/(1 + 3 \sqrt{57})^{1/3} + (1 + 3
  \sqrt{57})^{1/3})$, Lemma \ref{lem:(0,2)-mining} establishes that
  $(0,2)$ is not a mining state. But then from Lemma
  \ref{lem:bounds-play}, neither state $(0,1)$ is a mining state, or
  equivalently, \frontier\ is the best-response strategy of Miner 1.
\end{proof}

\subsection{Upper bound}

The main theorem of this section, Theorem \ref{thm:lower-bound}, shows
that if each miner $i$ has relative computational power $p_i<0.361$,
\frontier\ is a Nash equilibrium. On the other hand, it is intuitively
clear that if a miner has computational power close to $1/2$, he will
have some advantage if he does not play \frontier, against miners who
play \frontier. Recall that $h_0$ denotes the maximum relative computational power
of each miner for which \frontier\ is a Nash equilibrium. Experimental
results stated in Table \ref{tab:threshold-d}, based on computing the potential $\f$, show
that $h_0=0.418$.
\begin{table}[ht]
\centering
\caption{The threshold for different values of $d$.}
{
\begin{tabular}{l|cc}
 &&  threshold \\\hline \hline
$d=2$ && 0.5\\
$d=3$ && 0.454\\
$d=5$ && 0.432\\
$d=10$&& 0.422\\
$d=15$&& 0.42\\
$d=\infty$&& 0.418\\
  \hline
\end{tabular}}
\label{tab:threshold-d}\end{table}

Our work in this section is to estimate the threshold $h_0$. Providing
rigorous lower bounds for $h_0$---as we do in Theorem
\ref{thm:lower-bound}---does not appear to be easy since it involves a
non-trivial Markov decision process. However, it is not too hard to
obtain good upper bounds of $h_0$ as it suffices to come up with a
mining strategy that has expected gain $\g$ greater than $p$. Here we
provide a simple such upper bound which can be directly extended to
get an upper bound of $h_0$ close to $0.418$.

\begin{theorem}
When Miner 2 plays \frontier, the best response strategy for Miner 1 is not \frontier\ when $p\geq 0.455$.
\end{theorem}
\begin{proof}
It suffices to consider a fixed mining strategy $(M, s)$ that has expected gain per step $\g$ greater than $p$. We consider a mining strategy truncated at $d=3$,
i.e., the miner capitulates at every state $(a,b)$ with $b\geq 3$.

We select $M=\{(0,0),(0,1),(1,1),(1,2),(2,2)\}$ and $s=1$ and define the potential for this strategy as follows:
\begin{align*}
\f(0,0)&=0, \\
\f(0,1)&=\frac{\g-p}{1-p},\\
\f(a,b)&=\f(0,1), & \text{for every $(a,b)\in C$} \\
\f(a,b) &= a-\g, & \text{when $a=b+1$} \\
\f(2,2)&=p\f(3,2)+(1-p)\f(0,1)-\g(1-p),\\
\f(1,2)&=p\f(2,2)+(1-p)\f(0,1)-\g(1-p),\\
\f(1,1)&=p\f(2,1)+(1-p)\f(1,2)-\g(1-p).
\end{align*}
We need to select $\g$ and verify that this is the correct
potential. In particular, we need to verify that for all mining states $(a,b)\in M$ we have $\f(a,b)\geq \f(0,1)$, which holds when
$$\g=\frac{p^2(2+2p-5p^2+2p^3)}{1-p^2+2p^3-p^4}.$$
We can also verify that for $p\geq 0.455$, the expected gain $\g$
is strictly greater than $p$, which establishes that the strategy $(M,s)$ is a better response than \frontier.
\end{proof}
The strategy employed in the proof of the theorem is optimal for the truncated game at $d=3$ when $p\approx 0.455$. As we mentioned above, one can compute the optimal strategy for the finite games truncated at $d=4, 5$ and so on, to get better and better upper bounds. These bounds converge relatively quickly to $h_0=0.418$.


%% file: strategic_release.tex
\section{The Strategic-Release Game}
Similarly to the immediate-release case, we wish to identify conditions such that \frontier\ is a Nash equilibrium. To do so, we again assume
that all but one miner, say Miner 1, use the \frontier\ strategy, and
then examine the best response of Miner 1. Since all honest miners select the same block
to mine and release it immediately once it is mined, they essentially act
as one miner, and it is sufficient
to consider the case of two miners.

As in the immediate-release case, the tree created by the two miners has width 2. The trunk is permanently fixed and can be safely
ignored; hence, the situation is captured by the two branches of the
execution tree in addition to a bit of information regarding each
block in Miner 1's branch specifying whether or not Miner 1 has
released this block. In particular, this situation can be captured by a triple of numbers $a$, $a_\mathsf{r}$
and $b$, where $a$ and $b$ denote the number of blocks mined by Miner 1 and the honest miner, respectively (the length of the two branches), while $a_\mathsf{r}\leq a$ denotes the number of the released blocks on Miner's 1 branch (note that they are consecutive blocks starting from $(0,0)$). Since the honest miner immediately announces the mined blocks, all $b$ blocks on his branch are always released. Also, contrary to the immediate-release case, it can be $a>b+1$.

Under the assumption that  Miner 2 follows the \frontier\ strategy, we
know that if Miner 1 has released $a_\mathsf{r}\leq b$ blocks, then Miner 2
will not abandon his path, but if $a_\mathsf{r}>b$ then the honest miner will
immediately capitulate. So, without loss of generality we can assume
that if $a\leq b$ then $a_\mathsf{r}=a$ while if $a>b$, then $a_\mathsf{r}=b$, otherwise
the honest miner would have immediately abandoned his path, and the
game would be at state $(a-a_\mathsf{r},0)$. In other words, we can always
consider $a_\mathsf{r}=\min\{a,b\}$; hence we can capture the state of the game
by the tuple $(a,b)$ as in the immediate-release case.

Let $\gx_k(a,b)$ denote the expected gain of Miner 1 when the branch of the honest miner in the execution tree is extended by $k$ new
levels starting from an initial tree in which
the non-common parts of the paths have length $a$ and $b$, but Miner 1 has only released the first $\min\{a,b\}$ blocks on his path. It should
be intuitively clear that, in the long run, the expected gain per level
should be almost independent of the initial state. With this in mind,
we can write the expected gain as
$$\gx_k(a,b)=k\cdot \gxs + \ff(a,b),$$
for some constant $\gxs$ which is the expected gain per level in the
long run. The potential function $\ff(a,b)$ denotes the advantage
of Miner 1 for currently being in state $(a,b)$ (and having released $\min\{a,b\}$ blocks).

We can define $\gx_k(a,b)$ as follows:

\begin{itemize}
\item If $a\leq b$, Miner 1 has two options: to capitulate or to
  mine. In the latter case, the next state will be $(a+1,b)$ with
  probability $p$, and $(a,b+1)$ with the remaining probability.
\item If $a>b$, Miner 1 has one additional option to the previous
  case: he can release an additional block and lead the game to state
  $(a-b-1,0)$. When this happens, Miner 2 who plays \frontier\
  capitulates. Note that we allow Miner 1 to repeatedly release
  blocks, which is equivalent to allowing him to release any number of
  blocks.
\end{itemize}

From the above consideration, we get the following optimal gain
for Miner 1:
\begin{align}\nonumber
\gx_k(a,b)=&\max\left\{\max_{s=0,\ldots,b-1} \gx_{k}(0,s), p\cdot \gx_k(a+1, b)
            +
            (1-p)\cdot \gx_{k-1}(a, b+1), \right.\\ \label{rec:gx}
	& \left.  \textcolor{white}{\max_{s=0,\ldots,b-1}}\gx_{k-1}(a-b-1,0)+b+1
          \right\},
\end{align}
where the last term inside the $\max$ applies only when $a\geq b+1$;
equivalently we can define $\gx_k(a,0)=-\infty$ for any $k$ when $a<0$. The base
case of the recurrence is $\gx_0(a,b)=0$.

As in the case of immediate-release, we define a potential $\ff$ from
$$\gx_k(a,b)=k \gxs+\ff(a,b),$$
when $k$ tends to infinity. We also note that the definition above does not take $d$ into account (considers $d=\infty$). This is without loss of generality for the proof of Theorem \ref{thm:strategic-lower-bound} as proving the result for $d=\infty$ is stronger (as Miner 1 has a superset of available winning paths than for any constant $d$).

The recurrence for the potential is
\begin{align}\nonumber
  \ff(a,b)&=\max\left\{\max_{s=0,\ldots,b-1} \ff(0,s),p\cdot \ff(a+1,b)
            +
            (1-p)\cdot \ff(a,b+1)-\gxs \cdot(1-p),\right.\\ \label{rec:a>b}
	& \left. \textcolor{white}{\max_{s=0,\ldots,d-1}}\ff(a-b-1,0)+b+1-\gxs
          \right\},
\end{align}
where again we define $\ff(a,0)=-\infty$ when $a<0$. We also fix the
value $\ff(0,0)=0$.

\subsection{\frontier\ is a NE iff \boldmath$p\leq \hat h_0$, where $
  \hat h_0\geq 0.308$}

In this section we will show the following theorem.

\begin{theorem} \label{thm:strategic-lower-bound} In the strategic-release model, \frontier\ is a Nash equilibrium when every miner $i$ has relative computational power
  $p_i\leq 0.308$ (root of the polynomial $p^3- 6 p^2+ 5 p - 1$).
\end{theorem}

Theorem~\ref{thm:lower-bound} established that for any $p\leq 0.361$
there exists a potential $\f$ for the immediate-release model such
that $\f(0,0)=\f(0,1)$ and $\g=p$. The main idea of the proof is to
extend this potential to states $(a,b)$ such that $a>b+1$. This is
possible, for the following reasons:
\begin{itemize}
\item for states $(b+1,b)$, when $p$ is small and Miner 1 is one block
  ahead of Miner 2, there is a high risk if he does not release the
  block; the risk is that Miner 2 can mine one extra block and then
  Miner 1 may end up in a stale branch.
\item for states $(a,b)$ with $a>b+1$, when Miner 1 is at least two
  blocks ahead, it is safe not to release any blocks until Miner 2
  catches up to within distance one.
\end{itemize}
With this in mind, we define the following potential
\begin{equation*}
\sol(a,b)=\left\{\begin{array}{ll}
\f(a,b), &\text{ when } a\leq b,\\
a\lambda-b\mu-c, &\text{ otherwise}
\end{array}\right.
\end{equation*}
where $\lambda= \frac{(1-p)^2}{1-2p}$, $\mu=\frac{p^2}{1-2p}$ and
$c=\frac{p(1-p)}{1-2p}$. The parameters of the second part are chosen
so that $\sol(a,b)=p\sol(a+1,b)+(1-p)\sol(a,b+1)-p (1-p)$. It is
important to notice that the two pieces fit together nicely in the
sense that for $a=b+1$ we have that
$\sol(a,b)=\f(a,b)=a\lambda-b\mu-c$, so we could use the following
equivalent definition for $\sol$:
\begin{equation*}
\sol(a,b)=\left\{\begin{array}{ll}
\f(a,b), &\text{ when } a\leq b+1,\\
a\lambda-b\mu-c, &\text{ otherwise.}
\end{array}\right.
\end{equation*}
To ease the presentation, we will use the following notation, for the
potential of states $(a,b)$ with $a\geq b+1$: (a) $\sol_M$ when Miner 1 continues to mine, (b) $\sol_R$ when Miner 1 releases one more block and the other (honest) miner capitulates, and (c) $\sol_C$ when Miner 1 capitulates.
\begin{align*}
  \sol_M(a,b) &= p \sol(a+1,b)+(1-p)\sol(a,b+1)-p(1-p)\\
  \sol_R(a,b) &= \sol(a-b-1,0)+b+1-p \\
  \sol_C(a,b) &= 0.
\end{align*}
Note that when Miner 1 capitulates, he starts mining at one of the states $(0,s)$ where $s\leq b-1$, where by assumption the potential is $0$.
To prove the theorem, we first show that $\sol$ satisfies the
recurrence of the strategic-release potential \eqref{rec:a>b} when
$\gxs=p$. Equivalently, using the above notation
\begin{align}
  \sol(a,b)&=\max(\sol_M(a,b),\sol_R(a,b),\sol_C(a,b)) . \label{rec:a>b:shorthand}
\end{align}
\begin{lemma}\label{lem:strategic-potential-solution}
  The potential $\sol$ satisfies the recurrence \eqref{rec:a>b} when
  $\gxs=p$ and $p\leq 0.308$ (root of the polynomial
  $p^3- 6 p^2+ 5 p - 1$).
\end{lemma}
\begin{proof}
We break down the proof into three distinct claims:
\begin{claim}
  For states $(a,b)$ with $a < b+1$, $$\sol(a,b)=\max(\sol_M(a,b),\sol_R(a,b),\sol_C(a,b)).$$
\end{claim}
Using the alternative definition of $\sol$, the claim holds trivially, since $\f$ satisfies it, and $\sol_R(a,b)=-\infty$.

\begin{claim}
  For states $(a,b)$ with $a>b+1$, $$\sol(a,b)=\sol_M(a,b)=\max(\sol_M(a,b),\sol_R(a,b),\sol_C(a,b)).$$
\end{claim}
This claim also follows directly: $\sol(a,b)=\sol_M(a,b)$ holds
by design, and it is easy to verify that $\sol_M$ gives the maximum of
the three values:
\begin{align*}
\sol_M(a,b) &> a\lambda-b\mu-c-\frac{p(1-p)}{1-2p} = \sol_R(a,b), \\
\sol_R(a,b) &\geq 0 = \sol_C(a,b).
\end{align*}

\begin{claim}
  For states $(b+1,b)$:
  $$\sol(b+1,b)=\sol_R(b+1,b)=\max(\sol_M(b+1,b),\sol_R(b+1,b),\sol_C(b+1,b)).$$
\end{claim}
Note first that $\sol_R(b+1,b)=b+1-p\geq 0 \geq \sol_C(b+1,b)$. The most
complicated part is to show that $\sol_R(b+1,b)\geq
\sol_M(b+1,b)$. To do this, we write
\begin{align*}
  \sol_M(b+1,b)&= p\sol(b+2,b)+(1-p)\sol(b+1,b+1)-p(1-p) \\
               &= p((b+2)\lambda-b\mu-c)+(1-p)\f(b+1,b+1)-p(1-p).
\end{align*}
From Corollary \ref{cor:bound-restriction} for the immediate-release
case, we get the bound $\f(b+1,b+1)\leq (b+1) p/(1-p)$. By
substituting this we bound $\sol_M(b+1,b)$ from above by
$2 b p + p (2 - 4 p + p^2)/(1 - 2 p)$.  We want this to be at most
equal to $\sol_R(b+1,b)=b+1-p$ which gives the following inequality,
$$ b (1 - 2 p)^2 + 1 - 5 p + 6 p^2 - p^3 \geq 0.$$
Since $b\geq 0$, it suffices to have $1 - 5 p + 6 p^2 - p^3 \geq 0$
which holds for $p\leq 0.308$ (root of the polynomial
$p^3- 6 p^2+ 5 p - 1$).
\end{proof}

We now present the main result of this section based on the previous lemma.
\begin{proof}[Proof of Theorem \ref{thm:strategic-lower-bound}]
  The above lemma implies that for every state $(a,b)$,
  \begin{equation}
    \label{eq:4}
    \gx_k(a,b) \leq k\cdot p +\sol(a,b).
  \end{equation}
  Intuitively, $\sol(a,b)$ can only overestimate the optimal potential
  when $a\geq b+1$, even when $\gxs>p$. We prove this using induction on $k$, $b$, and $a$. This is possible because Recurrence
  \eqref{rec:gx} of $\gx_k$ depends on $\gx_k(0,s)$ for $s<b$ (for $b>0$), on
  $g_k(a+1,b)$, and/or  $\gx_{k-1}$.

For the outer induction on $k$, the base case $k=0$ is
  trivial since $\gx_0(a,b)=0$ and $\sol(a,b)\geq 0$. For a fixed $k$, we use double (strong) induction on $b$ and
  (backwards induction on) $a$. So, for fixed $k$ and $b$ we assume
  that the statement holds for all $k'<k$ and for all states $(a,b')$
  such that $b' < b$. Note that for the base case  $b=0$, the induction step does not use the inductive
hypothesis on $b$ since a rational miner never capitulates from a
state $(a,0)$.

  Since $\gx_k(a,b)$ cannot be bigger than $k+b$, there is a value of
  $a$ large enough such that $a\lambda-b\mu-c\geq k+b$; call this
  value $a_m(k,b)$. Statement \eqref{eq:4} holds for every
  $a\geq a_{m}(k,b)$ because
  $$\gx_k(a,b)\leq k+b \leq a\lambda-b\mu-c = \sol(a,b) \leq
  k\cdot p +\sol(a,b).$$ We can then use any value $\hat{a}\geq a_m(k,b)$ as base
  case of the backwards induction on $a$. More formally,
  \begin{align*}
  \gx_k(a,b)=&\max\left\{\max_{s=0,\ldots,b-1} \gx_{k}(0,s), p\, \gx_k(a+1, b)
            +
            (1-p)\, \gx_{k-1}(a, b+1), \right.\\
  & \left.  \textcolor{white}{\max_{s=0,\ldots,b-1}}\gx_{k-1}(a-b-1,0)+b+1
          \right\}\\
  \leq & kp +\max\left\{\max_{s=0,\ldots,b-1} \sol(0,s),  p\,\sol(a+1, b)
            +
            (1-p)\,\sol(a, b+1)-p(1-p), \right.\\
  & \left.  \textcolor{white}{\max_{s=0,\ldots,b-1}}\sol(a-b-1,0)+b+1-p
          \right\}\\
  =& kp +\sol(a,b),
  \end{align*}
where the inequality holds by our induction hypothesis, and the last
equality by Lemma \ref{lem:strategic-potential-solution}.

  Statement \eqref{eq:4} shows that the optimal gain per step
  cannot exceed $p$, which shows that \frontier\ is best response for
  Miner 1. The proof of the theorem is complete.
  \end{proof}

